\documentclass[10pt,conference]{IEEEtran}
\usepackage{graphicx}
\usepackage{subfigure}
\usepackage{color}
\usepackage{epsfig}
\usepackage{amssymb}
\usepackage{amsmath}
\usepackage{amsthm}
\usepackage{latexsym}

\newcommand{\bPP}[1]{{\mathtt{P}}_{#1}}
\newcommand{\tP}[1]{{\tilde{\mathtt{P}}}_{#1}}
\newcommand{\bP}[2]{{\mathtt{P}}_{#1}\left({#2}\right)}

\newcommand{\cX}{{\mathcal X}}
\newcommand{\cY}{{\mathcal Y}}
\newcommand{\cS}{{\mathcal S}}
\newcommand{\cU}{{\mathcal U}}
\newcommand{\cM}{{\mathcal M}}
\newcommand{\cP}{{\mathcal P}}
\newcommand{\cT}{{\mathcal T}}
\newcommand{\cV}{{\mathcal V}}

\newcommand{\ep}{\epsilon}
\newcommand{\bs}{\mathbf{s}}

\newtheorem{theorem}{Theorem}

\newtheorem{lemma}[theorem]{Lemma}
\theoremstyle{remark}
\theoremstyle{remark}
\newtheorem{remark}{Remark}
\theoremstyle{definition}
\newtheorem{definition}{Definition}

\begin{document}

\title{The Gelfand-Pinsker Channel: Strong Converse and Upper Bound for the Reliability Function}
\author{
\authorblockN{Himanshu Tyagi}
\authorblockA{Dept. of Electrical and Computer Engineering \\
        and \\
        Institute for Systems Research \\
        University of Maryland \\
        College Park, MD 20742, USA \\
        Email: tyagi@umd.edu} \and
\authorblockN{Prakash Narayan}
\authorblockA{Dept. of Electrical and Computer Engineering \\
        and \\
        Institute for Systems Research \\
        University of Maryland \\
        College Park, MD 20742, USA \\
        Email: prakash@umd.edu}}


\maketitle

\begin{abstract}
We consider a Gelfand-Pinsker discrete memoryless channel (DMC) model and provide
a strong converse for its capacity. The strong converse is then used to obtain an upper bound
on the reliability function. Instrumental in our proofs is a new technical lemma  which provides
an upper bound for the rate of codes with codewords that are conditionally typical over
large {\em message dependent subsets} of a typical set of state sequences. This
technical result is a nonstraightforward analog of a known result for a DMC without states
that provides an upper bound on the rate of a good code
with codewords of a fixed type (to be found in, for instance, the Csisz\'ar-K\"orner book).
\end{abstract}
\section{Introduction}
We consider a state dependent discrete memoryless channel (DMC),
in which the underlying state process is independent and identically distributed (i.i.d.)
with known probability mass function (pmf). The transmitter is provided access at the
outset to the entire state sequence prevailing during the transmission of a codeword.
The capacity of this DMC with noncausal channel state information (CSI) at the transmitter
was determined in \cite{GelPin80}. Known popularly as the Gelfand-Pinsker
channel, it has been widely studied for a broad range of applications
which include fingerprinting, watermarking, broadcast communication, etc.

In this paper, we are concerned with the {\it strong converse} for
this channel as well as its {\it reliability function}, i.e., the
largest exponential rate of decay, with block codeword length, of
the decoding error probability. Even for a DMC without states, the
reliability function is not fully characterized for all rates
below channel capacity. Our main contributions are the following.
First, we provide a strong converse for the capacity of the
Gelfand-Pinsker DMC model, that is of independent interest.
Second, using this strong converse, we obtain an upper bound for
the reliability function; the later constitutes a line of attack
described earlier (see, for instance, \cite{CsiKor81}).
Instrumental in the proofs of both is a new technical result which
provides an upper bound on the rate of codes with codewords that
are conditionally typical over large {\em message dependent}
subsets of a typical set of state sequences. This technical result
is a nonstraightforward analog of \cite[Lemma 2.1.4]{CsiKor81} for
a DMC without states; the latter provides a bound on the rate of a
good code with codewords of a fixed type.

\section{Preliminaries}
Consider a state dependent DMC $W: \cX \times \cS \rightarrow \cY$
with finite input, state and output alphabets
$\cX$, $\cS$ and $\cY$, respectively. The $\cS$-valued state process
$\{S_t\}_{t=1}^\infty$ is i.i.d. with known pmf $\bPP{S}$. The probability law
of the DMC is specified by
\begin{align}\nonumber
W^n(\mathbf{y}\mid \mathbf{x}, \bs) =&  \prod_{t=1}^nW(y_t\mid x_t, s_t),\\\nonumber&\mathbf{x}\in \cX^n\,, \bs\in \cS^n,\, \mathbf{y}\in \cY^n.
\end{align}
We consider the Gelfand-Pinsker model \cite{GelPin80} in which
the encoder possesses perfect CSI in a noncausal manner, i.e., the entire
state sequence prior to transmission. A $(M,n)$-code is a pair of
mappings $(f,\phi)$ where the encoder $f$ is a mapping
\begin{align}\nonumber
f:\cM \times \cS^n \rightarrow \cX^n
\end{align}
with $\cM = \{1,\dots,M\}$ being the set of messages, while the
decoder $\phi$ is a mapping
\begin{align}\nonumber
\phi: \cY^n \rightarrow \cM.
\end{align}
The rate of the code is $(1/n)\log{M}$.
The corresponding (maximum) probability of error is
\begin{align}\label{e2}\nonumber
e(f,\phi) =& \max_{m\in \cM}  \sum_{\bs\in \cS^n} \bPP{S}({\bs})\times\\&W^n((\phi^{-1}(m))^c\mid f(m,\bs),\bs)
\end{align}
where $\phi^{-1}(m) = \{\mathbf{y}\in \cY^n : \phi(\mathbf{y}) = m\}$ and $(\cdot)^c$
denotes complement. 

We restrict ourselves to the situation where the receiver has no CSI. When the receiver,
too, has (full) CSI, our results apply in a standard manner by considering an associated
DMC with augmented output alphabet $\cY \times \cS$.
\begin{definition}
Given $0<\ep<1$, a number $R>0$ is $\ep$-achievable if for every $\delta>0$
and for all $n$ sufficiently large, there exist $(M,n)$-codes $(f,\phi)$
with $(1/n)\log{M} > R - \delta$ and $e(f,\phi) < \ep$; $R$ is an
achievable rate if it is $\ep$-achievable for all $0<\ep<1$. The
supremum of all achievable rates is the capacity $C$ of DMC.
\end{definition}
For a random variable $U$ with values in a finite set $\cU$, let
$\cP$ denote the set of all pmfs $\bPP{USXY}$ on
$\cU\times\cS\times\cX\times\cY$ with
\begin{align}\label{e3}
X=h(U,S)
\end{align}
for some mapping $h$,
\begin{align}\label{e4}
U -\!\!\circ\!\!- &S,X -\!\!\circ\!\!- Y,\\P_{Y\mid X,S}&=W.
\end{align}
As is well-known \cite{GelPin80}
\begin{align}\nonumber
C = \max_{\cP }I(U\wedge Y) - I(U\wedge S).
\end{align}
When the receiver, too, has (full) CSI it is known \cite{Wol78} that
\begin{align}\nonumber
C = \max_{P_{X\mid S}} I(X\wedge Y\mid S).
\end{align}
\begin{definition}
The {\it reliability function} $E(R), \ R\geq 0$, of the DMC $W$
with noncausal CSI,
is the largest number $E\geq 0$ such that for every $\delta >0$ and
for all sufficiently large $n$, there exist $n$-length block
codes $(f,\phi)$ as above of rate greater than $R - \delta$ and
$e(f,\phi) \leq \exp{[ -n(E-\delta)] }$ (see for instance \cite{CsiKor81}).
\end{definition}

For a given pmf $\tP{SX}$ on $\cS \times \cX$, denote by
$\cP(\tP{SX},W)$ the subset of $\cP$ with $\bPP{SX} = \tP{SX}$.
\section{Statement of Results}
An upper bound for the reliability function $E(R)$, $0<R<C$,
of a DMC without states, is derived in \cite{CsiKor81} using a strong
converse for codes with codewords of a fixed type. For a state dependent
DMC with {\em causal} CSI at the transmitter and no receiver CSI, a strong converse
is given in \cite{Wol78}. An analogous result is not available for the case of
noncausal transmitter CSI. For the latter situation, the following
key lemma serves, in effect, as an analog of \cite[Corollary 2.1.4]{CsiKor81}
and gives an upper bound on the rate of codes with codewords that
are conditionally typical over large {\it message dependent} subsets
of the typical set of state sequences. We note that a direct extension
of \cite[Corollary 2.1.4]{CsiKor81} would have entailed a claim
over a subset of typical state sequences {\it not depending} on
the transmitted message; however, its validity is unclear.

For a DMC without states, the result in \cite[Corollary
2.1.4]{CsiKor81} provides, in effect, an image size
characterization of a good codeword set;  this does not involve
any auxiliary rv. In the same spirit, our key technical lemma
below provides an image size characterization for good codeword
sets for the noncausal DMC model, which now involves an auxiliary
rv.

\begin{lemma}\label{l_KL}
Let $\ep,\tau>0$ be such that $\ep + \tau<1$. Given a pmf $\tP{S}$
on $S$ and conditional pmf ${\tilde{P}}_{X\mid S}$, let $(f,\phi)$
be a $(M,n)$-code as above. For each $m \in \cM$, let $A(m)$
be a subset of $\cS^n$ which satisfies the following conditions
\begin{align}\label{e6}
A(m) \subseteq \cT^n_{[\tP{S}]},\\\label{e7} \|A(m)\| \geq
\exp{\left[n\left(H(\tP{S}) -
\frac{\tau}{6}\right)\right]},\\\label{e8} f(m,\bs) \in
\cT^n_{[\tP{X\mid S}]}(\bs),\,\,\,\,\,\,\,\, \bs \in A(m).
\end{align}
Furthermore, let $(f,\phi)$ satisfy one of the following two conditions
\begin{subequations}
\begin{align}\label{e9a}
W^n(\phi^{-1}(m)\mid f(m,\bs), \bs)&\geq 1- \ep,\,\,\,\,\,\,\,\, \bs\in A(m),
\\ \nonumber
\frac{1}{\|A(m)\|} \sum_{\bs\in A(m)}W^n(\phi^{-1}(m)&\mid f(m,\bs), \bs) \\\label{e9b}&\geq 1- \ep.
\end{align}
\end{subequations}
Then, for\footnote{In our assertions,
we indicate the validity of a statement ``for all
$n \geq N(.)$'' by showing the explicit dependency of $N$; however the
standard picking of the ``largest such $N$'' from (finitely-many)
such $N$s is not indicated.}
$n \geq N(\|\cX\|,\|\cS\|,\|\cY\|,\tau,\ep)$, it holds that
\begin{align}\nonumber
\frac{1}{n}\log{M} \leq I(U \wedge Y) - I(U \wedge S)
\end{align}
where $\bPP{USXY} \in \mathcal{P}(\tP{S}\tP{X\mid S},W)$.
\end{lemma}
\vspace{0.5cm}
This lemma plays an instrumental role in proving the following
two main results.

\begin{theorem}\label{t_SC}
{\it(Strong Converse)} Given $0<\ep<1$ and a sequence of $(M_n,n)$
codes $(f_n,\phi_n)$ with $e(f_n,\phi_n) < \ep$, it holds that
\begin{align}\nonumber
\limsup_n \frac{1}{n} \log{M_n} \leq C.
\end{align}
\end{theorem}
\vspace{0.5cm}

\begin{theorem}\label{t_SP}
({\it Sphere Packing Bound}) Given $\delta>0$, for $0< R < C$, it holds that
\begin{align}\nonumber
E(R) \leq E_{SP}(1+\delta) + \delta,
\end{align}
where
\begin{align}\label{e10}
E_{SP} =  \min_{\tP{S}}\max_{\tP{X\mid S}}\min_{V\in
\cV(R,\tP{S}\tP{X\mid S})} &\big[D(\tP{S}\|\bPP{S}) \\\nonumber&+
D(V\|W\mid \tP{S}\tP{X\mid S} )\big]
\end{align}
with
\begin{align}\nonumber
\mathcal{V}(R,\tP{SX}) = \big\{&V:\cX\times\cS \rightarrow \cY:
\\\nonumber& \max_{P_{USXY}\in \mathcal{P}(\tP{SX},V)} I(U \wedge Y) - I(U \wedge S)
< R \big\}.
\end{align}
\end{theorem}

\noindent
\begin{remark}
For the case when the receiver, too, possesses (full) CSI, the sphere
packing bound above coincides with that obtained earlier in \cite{Har01} for this case.
\end{remark}
\begin{remark}
In (\ref{e10}), the terms $D(\tP{S}\|\bPP{S})$ and $D(V\|W\mid
\tP{S}\tP{X\mid S})$ account, respectively, for the shortcomings
of a given code for corresponding ``bad" state pmf and ``bad"
channel.
\end{remark}
\section{Proofs of Results}
We provide below the proofs of Lemma \ref{l_KL} and Theorems \ref{t_SC} and \ref{t_SP}.

\vspace{0.1in}
\noindent
\begin{proof}[Proof of Lemma \ref{l_KL}] 

Our proof below is for the case when (\ref{e9a}) holds; the case
when (\ref{e9b}) holds can be proved similarly with minor
modifications. Specifically, in the latter case, we can find
subsets $A'(m)$ of $A(m)$, $m \in \cM$, that satisfy
(\ref{e6})-(\ref{e8}) and (\ref{e9a}) for some $\epsilon', \tau'
0$ with $\epsilon' + \tau' < 1$ for all $n$ sufficiently large.

Set
\begin{align}\nonumber
B(m) = \{(f(m,\bs),\bs) \in \cX^n\times\cS^n: \bs \in A(m) \},\, m\in \cM.
\end{align}
Let $\tP{Y} = \tP{SX} \circ W$ be a pmf on $\cY$ defined by
\begin{align}\nonumber
\tP{Y}(y)=  \sum_{s,x} \tP{SX}(s,x)W(y\mid x,s),\,  \ y\in \cY.
\end{align}
Consequently,
\begin{align}\label{e10a}
W^n(\cT^n_{[\tP{Y}]}\mid f(m,\bs),\bs)> \ep +
\tau,\,\,\,\,\,\,\,\,\,\bs\in A(m),
\end{align}
for all $n \geq N(\|\cX\|,|\cS\|,|\cY\|,\tau,\ep)$ (not depending on $m$
and $\bs$ in $A(m)$). Denoting
\begin{align}\nonumber
C(m) = \phi^{-1}(m)\cap \cT^n_{[\tP{Y}]},
\end{align}
we see from (\ref{e9a}) and (\ref{e10a}) that
\begin{align}\nonumber
W^n(C(m)\mid f(m,\bs),\bs) > \tau >0,\,\,\,\,\,\,\,\, (f(m,\bs),\bs)\in B(m),
\end{align}
so that
\begin{align}\nonumber
\|C(m)\| \geq g_{W^n}(B(m),\tau),
\end{align}
where $g_{W^n}(B(m),\tau)$ denotes the smallest cardinality of
a subset $D$ of $\cY^n$ with
\begin{align}\label{e10b}
W^n(D\mid (f(m,\bs),\bs)) > \tau,\,\,\,\,\,\,\,\,(f(m,\bs),\bs)\in
B(m).
\end{align}
With $m_0 = \mathtt{arg} \min_{1\leq m \leq M} \|C(m)\|$, we have
\begin{align}\nonumber
M\|C(m_0)\| \leq \sum_{m = 1}^{M}\|C(m)\| = \|\cT^n_{[\tP{Y}]}\|
\leq \exp{n\bigg(H(\tP{Y}) + \frac{\tau}{6}\bigg)}.
\end{align}
Consequently,
\begin{align}\label{e11}
\frac{1}{n}\log M \leq H(\tP{Y}) + \frac{\tau}{6} -
\frac{1}{n}\log \mathtt{g_{W^n}}(B(m_0),\tau).
\end{align}
Define a stochastic matrix $V: \cX\times \cS \rightarrow \cS$  with
\begin{align}\nonumber
V(s'\mid x,s) = \mathbf{1}(s' = s),
\end{align}
and let $\mathtt{g_{V^n}}$ be defined in a manner analogous to
$\mathtt{g_{W^n}}$ above with $\cS^n$ in the role of $\cY^n$ in
(\ref{e10b}). For any $m \in \cM$ and subset $E$ of $\cS^n$,
observe that
\begin{align}\nonumber
V^n(E\mid f(m,\bs),\bs) = \mathbf{1}(s \in E),\,\,\,\,\,\,\,\, \bs
\in \cS^n.
\end{align}
In particular, if $E$ satisfies
\begin{align}\label{e10c}
V^n(E\mid f(m,\bs),\bs) > \tau,\,\,\,\,\,\,\,\, \bs \in A(m),
\end{align}
it must be that $A(m) \subseteq E$, and since $E = A(m)$ satisfies
(\ref{e10c}), we get that
\begin{align}\label{e12}
\|A(m)\| = \mathtt{g_{V^n}}(B(m),\tau)
\end{align}
using the definition of $B(m)$. Using the image size
characterization \cite[Theorem 3.3.11]{CsiKor81}, there exists an
auxiliary rv $U$ and associated pmf $\bPP{USXY} = \bPP{U\mid
SX}\tP{SX}W$ such that
\begin{align}\label{e13}\nonumber
\left|\frac{1}{n}\log  \mathtt{g_{V^n}}(B(m_0),\tau) - H(S|U) - t\right| &< \frac{\tau}{6},\\
\left|\frac{1}{n}\log  \mathtt{g_{W^n}}(B(m_0),\tau) - H(Y|U) -
t\right| &< \frac{\tau}{6},
\end{align}
where $0\leq t \leq\min\{I(U \wedge Y),I(U \wedge S)\}$. Then,
using (\ref{e11}), (\ref{e12}), (\ref{e13}) we get
\begin{align}\nonumber
\frac{1}{n}\log M \leq I(U \wedge Y) + H(S\mid U) -
\frac{1}{n}\log \|A(m_0)\| + \frac{\tau}{2},
\end{align}
which by (\ref{e7}) yields
\begin{align}\nonumber
\frac{1}{n}\log M \leq I(U \wedge Y) - I(U \wedge S)  + \tau.
\end{align}
In (\ref{e13}), $\bPP{USXY}$ belongs to
$\mathcal{P}(\tP{S}\tP{X\mid S},W)$ but need not satisfy
(\ref{e3}). Finally, the asserted restriction to $\bPP{USXY} \in
\mathcal{P}(\tP{S}\tP{X\mid S},W)$ follows from the convexity of
$I(U\wedge Y) - I(U\wedge S)$ in $\bPP{X\mid US}$ for a fixed
$\bPP{US}$ (as observed in \cite{GelPin80}).
\end{proof}

\vspace{0.1in}
\noindent\begin{proof}[Proof of Theorem \ref{t_SC}]

Given $0<\ep<1$ and a $(M,n)$-code $(f,\phi)$ with $e(f,\phi)\leq
\ep$, the proof involves the identification of sets $A(m)$, $m \in
\cM$, satisfying (\ref{e6})-(\ref{e8}) and (\ref{e9a}). The
assertion then follows from Lemma \ref{l_KL}. Note that $e(f,\phi)
\leq \ep$ implies
\begin{align}\nonumber
\sum_{\bs\in \cS^n}\bP{S}{\bs}W^n(\phi^{-1}(m)\mid f(m,\bs), \bs) \geq 1- \ep
\end{align}
for all $m \in \cM$. Since $\bP{S}{\cT^n_{[\bPP{S}]}} \rightarrow 1$ as $n\rightarrow \infty$,
we get that for every $m \in \cM$,
\begin{align}\label{e14}\nonumber
\bP{S}{\bigg\{\bs\in \cT^n_{[\bPP{S}]} : W^n(\phi^{-1}(m)\mid f(m,\bs), \bs) > \frac{1- \ep}{2}\bigg\}} \\
\geq \frac{1-\ep}{3}
\end{align}
for all $n \geq N(\|\cS\|,\ep)$. Denoting the set $\{\cdot\}$
in (\ref{e14}) by $\hat{A}(m)$, clearly for every $m \in \cM$,
\begin{align}\nonumber
W^n(\phi^{-1}(m)\mid f(m,\bs),\bs)\geq
\frac{1-\ep}{2},\,\,\,\,\,\,\,\,\bs \in \hat{A}(m),
\end{align}
and
\begin{align}\nonumber
\bP{S}{\hat{A}(m)} \geq \frac{1- \ep}{3}
\end{align}
for $n \geq N(\|\cS\|,\ep)$, whereby for an arbitrary $\delta>0$, we get
\begin{align}\nonumber
\|\hat{A}(m)\|\geq \exp{[n(H(\bPP{S})-\delta)]}
\end{align}
for $n \geq N(\|\cS\|,\delta)$. Partitioning $\hat{A}(m)$, $m \in \cM$,
into sets according to the (polynomially many) conditional
types of $f(m,\bs)$ given $\bs$ in $\hat{A}(m)$, we obtain
a subset $A(m)$ of $\hat{A}(m)$ for which
\begin{align}
\nonumber f(m,\bs) &\in \cT^n_m(\bs),\,\,\,\,\,\,\,\,\bs\in A(m),\\
\nonumber \|A(m)\| &\geq\exp{[n(H(\bPP{S})-2\delta)]},
\end{align}
for $n\geq N(\|\cS\|,\|\cX\|,\delta)$, where $\cT^n_m(\bs)$
represents a set of those sequences in $\cX^n$ that have the
same conditional type (depending only on $m$).

Once again, the polynomial size of such conditional
types yields a subset $\cM'$ of $\cM$ such that $f(m,\bs)$
has a fixed conditional type (not depending on $m$)
given $\bs$ in $A(m)$, and with
\begin{align}\nonumber
\frac{1}{n}\log{\|\cM'\|} \geq \frac{1}{n}\log{M} - \delta
\end{align}
for all $n\geq N(\|\cS\|,\|\cX\|,\delta)$. Finally, the strong converse follows
by applying Lemma \ref{l_KL} to the subcode corresponding
to $\cM'$ and noting that $\delta>0$ is arbitrary.
\end{proof}

\vspace{0.1in}
\noindent
\begin{proof}[Proof of Theorem \ref{t_SP}]

Consider sequences of type $\tP{S}$ in $\cS^n$. Picking
$\hat{A}(m) = \cT^n_{\tP{S}}$, $m\in \cM$, in the proof of Theorem
\ref{t_SC}, and following the arguments therein to extract the
subset $A(m)$ of $\hat{A}(m)$, we have for a given $\delta>0$ that
for $n\geq N(\|\cS\|,\|\cX\|,\delta)$, there exists a subset
$\cM'$ of $\cM$ and a fixed conditional type, say $\tP{X\mid S}$
(not depending on $m$), such that for every $m \in \cM'$,
\begin{align}
\nonumber A(m) &\subseteq \hat{A}(m) = \cT^n_{\tP{S}},\\
\nonumber \|A(m)\|&\geq \exp{[n(H(\tP{S}) - \delta)]},\\
\nonumber f(m,\bs) &\in \cT^n_{\tP{X\mid S}}(\bs),\,\,\,\,\,\,\,\,\,\,\,\,\,\,\,\, \bs \in A(m),\\
\nonumber \frac{1}{n}\log{\|\cM'\|} &\geq R - \delta.
\end{align}
Then for every $V \in \cV(R,\tP{S}\tP{X\mid S})$, we obtain
using Lemma \ref{l_KL} (in its version with condition (\ref{e9b})),
that for every $\delta'>0$, there exists $m\in\cM'$
(possibly depending on $\delta'$ and $V$) with
\begin{align}
\nonumber \frac{1}{\|A(m)\|}\sum_{\bs\in A(m)}V^n((\phi^{-1}(m))^c\mid f(m,\bs),\bs) \geq 1 - \delta'
\end{align}
for all $n \geq N(\|\cS\|,\|\cX\|,\|\cY\|,\delta')$. For this
$m$, apply \cite[Theorem 2.5.3, (5.21)]{CsiKor81} with the choices
\begin{align}
\nonumber Z &= \cY^n\times A(m),\\
\nonumber S &= (\phi^{-1}(m))^c\times A(m),\\
\nonumber Q_1(\mathbf{y},\bs) &= \frac{V^n(\mathbf{y} \mid
f(m,\bs),\bs)}{\|A(m)\|},
\end{align}
\newpage
\begin{align}
\nonumber Q_2(\mathbf{y},\bs) &= \frac{W^n(\mathbf{y} \mid
f(m,\bs),\bs)}{\|A(m)\|},
\end{align}
for $(\mathbf{y},\bs) \in Z$, to obtain
\begin{align}
\nonumber \frac{1}{\|A(m)\|} \sum_{\bs\in A(m)}&W^n((\phi^{-1}(m))^c\mid f(m,\bs), \bs) \\
\nonumber &\geq \exp{\left[-\frac{nD(V\| W\mid \tP{X\mid S}\tP{S}) + 1}{1 - \delta'}\right]}.
\end{align}
Finally,
\begin{align}
\nonumber e(f,\phi)& \geq \sum_{\bs\in A(m)}\bP{S}{\bs}W^n({(\phi^{-1}(m))}^{c}\mid f(m,\bs), \bs)
\end{align}
\begin{align}
\nonumber &\geq \exp[-n(D(\tP{S}\| \bPP{S}) \\
\nonumber&\qquad \qquad+ D(V\| W\mid \tP{X\mid S}\tP{S})(1 + \delta) + \delta)]
\end{align}
for $n \geq N(\|\cS\|,\|\cX\|,\|\cY\|,\delta,\delta')$,
whereby it follows that
\begin{align}
\nonumber\lim \sup_n&-\frac{1}{n}\log{e(f,\phi)}\\
\nonumber &\leq \min_{\tP{S}}\max_{\tP{X\mid S}}\min_{V\in \cV(R,\tP{S}\tP{X\mid S})}[D(\tP{S}\| \bPP{S}) \\
\nonumber&\qquad\qquad+ D(V\| W\mid \tP{X\mid S}\tP{S})(1 + \delta) + \delta]
\end{align}
for every $\delta>0$.
\end{proof}

\vspace{0.1in}
\section*{Acknowledgements}

The authors thank Shlomo Shamai for helpful comments.

The authors' work was supported by the U.S. National Science
Foundation under Grant ECS0636613.



\vspace{0.1in}

\begin{thebibliography}{1}

\bibitem{GelPin80}
S.~I. Gelfand and M.~S. Pinsker.
\newblock Coding for channels with random parameters.
\newblock {\em Problem of Control and Information Theory}, 9(1):19--31, 1980.

\bibitem{CsiKor81}
I.~Csisz{\'a}r and J.~K{\"o}rner.
\newblock {\em Information theory: coding theorems for discrete memoryless
  channels}.
\newblock Academic Press, 1981.

\bibitem{Wol78}
J.~Wolfowitz.
\newblock {\em Coding theorems of information theory}.
\newblock New York:Springer-Verlag, 1978.

\bibitem{Har01}
M.~E. Haroutunian.
\newblock New bounds for ${E}$-capacities of arbitrary varying channel and
  channel with random parameter.
\newblock {\em Mathematical Problems of Computer Science}, 22:44--59, 2001.

\end{thebibliography}
\end{document}